\documentclass[a4paper,12pt]{article}

\usepackage{amsmath, amsthm, amsfonts, amssymb, dsfont, color, accents}
\usepackage{amsmath}
\usepackage{amssymb}
\usepackage{amsfonts,dsfont,enumerate}

\theoremstyle{plain}
\newtheorem{theorem}{Theorem}[section]

\newtheorem{lemma}{Lemma} %[section]

\theoremstyle{remark}

\newcommand{\lmb}{\lambda}
\newcommand{\Lmb}{\Lambda}

\newcommand{\e}{\varepsilon}
\newcommand{\g}{\gamma}

\renewcommand{\d}{\delta}

\newcommand{\p}{\partial}
\newcommand{\D}{\Delta}
 \newcommand{\E}{\mbox{\rm e}}
 
\newcommand{\vp}{\varphi}
\newcommand{\vk}{\varkappa}
\newcommand{\Odr}{\mathcal{O}}

\newcommand{\iu}{i}

\newcommand{\Op}{\mathcal{H}}
\newcommand{\Res}{\mathcal{R}}
\newcommand{\I}{\mathrm{I}}

\usepackage{amsthm, amsmath, amsfonts, amssymb, accents}
\usepackage{graphicx}
\usepackage{srcltx}
\usepackage{dsfont}

\DeclareMathOperator{\spec}{\sigma}
\DeclareMathOperator{\dist}{dist}

\DeclareMathOperator{\RE}{Re}

\begin{document}

% Title
\title{On the spectral stability of kinks in 2D Klein-Gordon
model with parity-time-symmetric perturbation}

%Authors,  affiliations address.
\author{Denis I. Borisov\,$^{a,b,c}$, Sergey V. Dmitriev$^{d,e}$}

\date{\empty}

\maketitle

\begin{center}

\small{
\begin{enumerate}[a)]

\item Institute of Mathematics CC USC RAS, Chernyshevsky Str.~112, Ufa, Russia, 450008
    
\item Bashkir State Pedagogical University, October Revolution Str.~3a, Ufa, Russia, 450000
    
\item University of Hradec Kr\'alov\'e, Rokitansk\'eho 62, 500 03, Hradec Kr\'alov\'e, Czech Republic
    
\item Institute for Metals Superplasticity Problems of RAS, Khalturin Str.,~39, Ufa, Russia, 450001
    
\item National Research Tomsk State University, 36 Lenin Street, Tomsk, Russia, 634050
\end{enumerate}}

\end{center}

%Abstract
\begin{abstract}
In a series of recent works by Demirkaya {\it et al.} stability analysis for the static kink solutions to the 1D continuous and discrete Klein-Gordon equations with a $\mathcal{PT}$-symmetric perturbation has been analysed. We consider the linear stability problem for the static kink in 2D Klein-Gordon field taking into account spatially localized $\mathcal{PT}$-symmetric perturbation. The perturbation is in the form of viscous friction, which does not affect the static solutions to the unperturbed Klein-Gordon equation. Small dynamic perturbation around the static kink solution is considered to formulate the linear stability problem. The effect of the small perturbation on the solutions to the corresponding eigenvalue problem is analysed. The main result is presented in the form of a theorem describing the behavior of the eigenvalues corresponding to the extended and localised eigenmodes as the functions of the perturbation parameter.
\end{abstract}

% Heading 1
\section{Introduction} \label{Sec1}

The observation that non-Hermitian operators can possess a purely real spectrum under parity-time ($\mathcal{PT}$) symmetry condition \cite{Bender1,Bender2,Bender3} turned out to be very fruitful for many physical applications such as quantum mechanics \cite{QM1,QM2} and optics \cite{Optics1,Optics2,Optics3}, with some impact even on the electrical \cite{Electricity1,Electricity2} and mechanical \cite{Mechanics,Mechanics1} systems.

Very recently Demirkaya {\it et al.} have proposed a $\mathcal{PT}$-symmetric perturbation for the 1D Klein-Gordon field \cite{KG1}. The perturbation is such that the properties of static solutions are not affected. In that work the authors have addressed the problem of static kink stability. It has been found that the static kink is stable (unstable) if placed in the lossy (gain) region \cite{KG1,DKKSS}. Existence and stability of the dynamical excitation such as the sine-Gordon breather have been investigated by Lu {\it et al.} \cite{Lu}. According to the results of their study, the (unstable) breather exists if located exactly at the center between the gain and loss sides of the perturbation, while it cannot survive being placed in the gain or loss region. Stability analysis has been performed for the periodic standing wave solutions for the coupled nonlinear Schr\"odinger equations \cite{DemirkayaPLA}. Parameters that ensure the stability of the studied solutions have been calculated.

Kevrekidis has developed a collective variable method to describe the dynamics of kinks in 1D Klein-Gordon field in the presence of $\mathcal{PT}$-symmetric perturbation \cite{K1}. In the later works this method was applied to the sine-Gordon \cite{Danial1} and $\phi^4$ \cite{Danial2,Danial3} models.

In the present study we extend the stability analysis of the static kinks by considering two-dimensional Klein-Gordon model with $\mathcal{PT}$-symmetric perturbation. In Sec. \ref{Sec2} we describe the 2D nonlinear Klein-Gordon equation with the $\mathcal{PT}$-symmetric perturbation and formulate the eigenvalue problem for the perturbation around static kink solution. Then, the main result describing the behaviour of the eigenvalues as the functions of small perturbation parameter is formulated in Theorem 1. The proof of the theorem is presented in Sec. \ref{Sec3}. A brief conclusion is given in Sec. \ref{Sec4}.

%Heading 1
\section{Motivation, problem, and main results} \label{Sec2}

We begin with the perturbed two-dimensional Klein-Gordon equation
\begin{equation}\label{a1}
u_{tt} -\D u + \e \gamma(x,y)u_t + f(u) = 0,\quad t>0,\quad (x,y)\in\mathds{R}^2.
\end{equation}
Here $\D$ stands for the Laplace operator $\D=\frac{\p^2}{\p x^2}+\frac{\p^2}{\p y^2}$, $(x,y)$ are the Cartesian coordinates in $\mathds{R}^2$, $\g(x,y)u_t$ is the perturbation term describing viscous gain and loss in the system. We assume that  function $\g$  is continuous in $\mathds{R}^2$, decays exponentially as $x^2+y^2\to\infty$, namely,
\begin{equation*}
|\g(x,y)|\leqslant Ce^{-a\sqrt{x^2+y^2}},\quad (x,y)\in\mathds{R}^2,
\end{equation*}
where $C$, $a$ are fixed positive constants, and satisfies one of the following oddness properties:
\begin{equation*}%\l%abel{2.2}
\g(-x,-y)=-\g(x,y)\quad\text{or}\quad \g(-x,y)=-\g(x,y)\quad\text{or}\quad \g(x,-y)=-\g(x,y).
\end{equation*}
Function $f$ is supposed to be continuously differentiable and such that the equation
\begin{equation*}%\l%abel{2.3}
-\phi_{xx}+f(\phi)=0
\end{equation*}
has a static kink solution $\phi=\phi(x)$.

We make the following assumptions for functions $f$ and $\phi$. Consider the operator
\begin{equation*}
\Op_0:=-\frac{d^2\hphantom{x}}{dx^2}+f'(\phi(x))\quad \text{in}\quad L_2(\mathds{R}),
\end{equation*}
and suppose that its essential spectrum is a subset of the half-line $[\Lmb_e,+\infty)$, and $\Lmb_e$ is the bottom of this essential spectrum. A typical behaviour of the function $f'(\phi)$ at infinity is
\begin{equation}\label{2.4}
\lim\limits_{x\to\pm\infty} f'(\phi(x))= \Lmb_e^\pm,
\end{equation}
where $\Lmb_e^\pm$ are some constants. Then the essential spectrum of $\Op_0$ is exactly $[\Lmb_e,+\infty)$ with $\Lmb_e:=\min\{\Lmb_e^-,\Lmb_e^+\}$. Particularly, for the sine-Gordon field one has
\begin{equation*}%\l%abel{a3}
f(u) = \sin(u), \quad \phi(x,y) = 4 \arctan(e^x),
\end{equation*}
and (\ref{2.4}) is satisfied with $\Lmb_e^\pm=1$.
For the $\phi^4$ model
\begin{equation*}%\%label{a5}
f(u) = 2(u^3 - u),\quad \phi(x,y) = \tanh(x),
\end{equation*}
and (\ref{2.4}) holds true with $\Lmb_e^\pm=4$.

In this paper we study the stability of kink $\phi$. Namely, we consider a small dynamical perturbation of the static kink
\begin{equation*}%\l%abel{a7}
u(x,y,t) = \phi(x,y) + \psi(x,y,t).
\end{equation*}
We substitute this sum into (\ref{a1}) and linearize the obtained equation w.r.t. $\phi$. The obtained linear equation reads as
\begin{equation*}
\vp_{tt}-\D\vp +\e\g\vp_t+f'(\phi)\vp=0,\quad t>0,\quad (x,y)\in \mathds{R}^2.
\end{equation*}
Then we separate spatial and time variables assuming that
\begin{equation*}
\vp(x,y,t)=e^{\lmb t} \psi(x,y).
\end{equation*}
It leads us to the following eigenvalue problem
\begin{equation}\label{2.5}
(-\D+f'(\phi)+\e\g\lmb + \lmb^2)\psi=0,\quad (x,y)\in\mathds{R}^2.
\end{equation}
Exactly this equation is the main object of our study in the present paper. We are interested in non-trivial solutions to this equation and the associated values of spectral parameter $\lmb$. Namely, we are interested in the eigenvalues and the resonances of the above equation. We define eigenvalues as $\lmb$, for which there exists a non-trivial solution to (\ref{2.5}) belonging to Sobolev space $W_2^2(\mathds{R}^2)$. We recall that space $W_2^2(\mathds{R}^2)$ is a subset of functions in $L_2(\mathds{R}^2)$ whose first and second generalized derivatives are also elements of $L_2(\mathds{R}^2)$.

A resonance is a value of $\lmb$, for which there exists a non-trivial solution to (\ref{2.5})  belonging to Sobolev space $W_{2,loc}^2(\mathds{R}^2)$ and behaving at infinity as
\begin{equation}\label{2.15}
\psi(x,y)=C_\pm(x)e^{-\sqrt{\Lmb_e+\lmb^2}|y|}+\ldots,\quad y\to\pm\infty,
\end{equation}
where $C_\pm(x)$ are some functions. In all our considerations below functions $C$ belong to $W_2^2(\mathds{R})$. We also recall that $W_{2,loc}^2(\mathds{R}^2)$ is a set of functions belonging to $W_2^2(Q)$ for each bounded domain $Q\subset \mathds{R}^2$.

We shall call a resonance simple if there exists just one associated non-trivial solution to (\ref{2.5}), (\ref{2.15}).

Our first lemma describes the essential spectrum of the operator $\Op:=-\D+f'(\phi(x))$ in $L_2(\mathds{R}^2)$.

\begin{lemma}\label{lm3.1}
The essential spectrum of $\Op$ is the semi-axis $[\Lmb_0,+\infty)$,
where $\Lmb_0$ is the bottom of the spectrum of $\Op_0$.
\end{lemma}

To formulate the main results, we need additional notations. Let $\Lmb_*<\Lmb_e$ be a discrete eigenvalue of operator $\Op_0$. Since this operator is one-dimensional,   eigenvalue $\Lmb_*$ is simple and by $\psi_*=\psi_*(x)$ we denote the associated real eigenfunction normalized in $L_2(\mathds{R})$.
We define $\vk_*:=\sqrt{\Lmb_*}$ and note that $\vk_*\geqslant 0$ if $\Lmb_*>0$ and it is pure imaginary with $\mathrm{Im}\,\vk_*>0$ if $\Lmb_*<0$. We let
%\begin{equation}\l%abel{3.22}
\begin{align*}
K_1:=\frac{1}{2} \int\limits_{\mathds{R}^2} \g(x,y)\psi_j^2(x) dxdy,
\quad 
K_2:=\frac{1}{2} \int\limits_{\mathds{R}^2} \g(x,y)\psi_j(x) U_*(x,y) dxdy,
\end{align*}
%\end{equation}
where $U_*$ is introduced as the solution to the equation
\begin{equation}\label{2.16}
(-\D+f'(\phi)-\Lmb_*)U_*=\g\psi_*-\psi_* \int\limits_{\mathds{R}^2} \psi_*^2\g dxdy,
\end{equation}
behaving at infinity as
\begin{equation}\label{2.17}
U_*(x,y)=\frac{\psi_*(x)}{2} \int\limits_{\mathds{R}^2} |y-s_2|\psi_*^2(s_1) \g(s_1,s_2) ds_1 ds_2+o(1),\quad y\to\pm\infty.
\end{equation}
In what follows we shall show that this problem is uniquely solvable.

By $\spec(\Op_0)$ we denote the spectrum of operator $\Op_0$.

Our main results reads as follows.

\begin{theorem}\label{th2.2}
1. As $\e\to+0$, each of the eigenvalues and resonances $\lmb_\e$ of equation (\ref{2.5}) either tends to infinity or the distance from $-\lmb_\e^2$ %this eigenvalue/resonance
to the set
$\spec(\Op_0)\setminus\{0\}$
tends to zero.

2. Suppose that $\Lmb_*>0$, $K_2>0$. Then for all sufficiently small $\e$  equation (\ref{2.5}) has exactly two eigenvalues $\lmb^\pm_\e$ converging to $\pm\iu\vk_*$ as $\e\to+0$. These eigenvalues are simple and have the asymptotics
\begin{equation}\label{3.24}
\begin{aligned}
&\lmb^\pm_\e=\pm\iu\vk_* \pm   \frac{\iu\vk_* K_1^2}{2} \e^2 + \Lmb_* K_1 K_2
\e^3+  \Odr(\e^4), && \text{if}\quad K_1\not=0,
\\
&\lmb^\pm_\e=\pm\iu\vk_* \mp  \iu\vk_* \Lmb_*K_2^2\e^4  +  \Odr(\e^5), && \text{if}\quad K_1=0.
\end{aligned}
\end{equation}

3. Suppose that $\Lmb_*>0$, $K_2<0$. Then for all sufficiently small $\e$  equation (\ref{2.5}) has exactly two resonances converging to $\pm\iu\vk_*$ as $\e\to+0$. These resonances are simple and have   asymptotics (\ref{3.24}).

4. Suppose that $\Lmb_*<0$, $K_1\not=0$. Then for all sufficiently small $\e$  equation (\ref{2.5}) has exactly one eigenvalue $\lmb_\e^+$ and one resonance $\lmb_\e^-$ converging to $\pm|\vk_*|$ as $\e\to+0$. These eigenvalue and resonance are simple and have the asymptotics
\begin{align*}
&\lmb^\pm_\e=\pm|\vk_*| \pm   \frac{ K_1^2|\vk_*|}{2} \e^2 + K_1 K_2 \Lmb_*
\e^3+  \Odr(\e^4),\quad \text{if}\quad K_1>0,
\\
&\lmb^\pm_\e=\mp|\vk_*| \mp   \frac{ K_1^2|\vk_*|}{2} \e^2 + K_1 K_2 \Lmb_*
\e^3+  \Odr(\e^4),\quad \text{if}\quad K_1<0.
\end{align*}

5.  Suppose that $\Lmb_*<0$, $K_1=0$, $K_2>0$. Then for all sufficiently small $\e$  equation (\ref{2.5}) has exactly two resonances  converging to $\pm|\vk_*|$ as $\e\to+0$. These resonances are simple and have the asymptotics
\begin{equation}\label{2.19}
\lmb^\pm_\e=\pm|\vk_*| \pm \frac{|\vk_*| \Lmb_*K_2^2}{2} \e^4 +  \Odr(\e^5).
\end{equation}

6.  Suppose that $\Lmb_*<0$, $K_1=0$, $K_2<0$. Then for all sufficiently small $\e$  equation (\ref{2.5}) has exactly two eigenvalues  converging to $\Lmb_*$ as $\e\to+0$. These eigenvalues are simple and have asymptotics (\ref{2.19}).
\end{theorem}

Let us discuss briefly this theorem. The spectrum of operator $\Op_0$ is a fixed set having the essential part $[\Lmb_e,+\infty)$ and finitely many discrete eigenvalues below $\Lmb_e$. And Statement~1 of Theorem~1 says that there are just two possible types of the behavior for the eigenvalues and the resonances of $\Op_\e$. The first option is just escaping to the infinity as $\e\to+0$. If this is not the case and an eigenvalue/resonance $\lmb_\e$ remains in a bounded domain of the complex plane, then $-\lmb_\e^2$ is located in a small neighbourhood of the set $\spec(\Op_0)\setminus\{0\}$. And this neighbourhood becomes smaller as $\e\to+0$. Point $0$ is exceptional in the sense that the eigenvalues/resonances do no tend to this point.

The essential spectrum of operator $\Op$ is wider than that of operator $\Op_0$, see Lemma~\ref{lm3.1}. Each discrete eigenvalue of $\Op_0$ is an internal point in the essential spectrum of $\Op$. Although being embedded into the essential spectrum, these points are exceptional. Namely, each such non-zero point generate either two eigenvalues or two resonances or a pair of an eigenvalue and a resonance. By generating we mean that there are two eigenvalues/resonances $\lmb_\e^\pm$ of equation (\ref{2.5}) such that $(\lmb_\e^\pm)^2\to -\Lmb_*$ as $\e\to+0$, where $\Lmb_*$ is a discrete eigenvalue of $\Op_0$. Statements~2-6 of Theorem~1 describes how to determine whether we have eigenvalues or resonances and provide also the leading terms of their asymptotic expansions. As we see, the type of the emerging spectral point is very sensitive to the signs of $\Lmb_*$, $K_1$, $K_2$. Statements~2-6 describe main possible cases, but not all possible ones. For instance, the case $K_1=K_2=0$ is not covered by Theorem~\ref{th2.2}. But our technique can be easily extended to such uncovered case. On the other hand, it requires additional technical bulky calculations. This is the reason why we do not include these cases in the present paper.

\section{Proof of main result} \label{Sec3}

In this section we prove Lemma~\ref{lm3.1} and Theorem~\ref{th2.2}.

\begin{proof}[Proof of Lemma~\ref{lm3.1}.]
We write the quadratic form for $\Op$ and estimate it from below by employing the definition of $\Lmb_0$:
\begin{equation}\label{3.9}
\begin{aligned}
(\Op\psi,\psi&)_{L_2(\mathds{R}^2)}=\|\nabla \psi\|_{L_2(\mathds{R}^2)}^2+(f'(\phi)\psi,\psi)_{L_2(\mathds{R}^2)}
\\
=&\int\limits_{\mathds{R}} \left(\|\nabla \psi(\cdot,y)\|_{L_2(\mathds{R})}^2 + \big(f(\phi')\psi(\cdot,y),\psi(\cdot,y)\big)_{L_2(\mathds{R})} \right)dy
\\
\geqslant& \Lmb_0 \int\limits_{\mathds{R}} \|\psi(\cdot,y)\|_{L_2(\mathds{R})}^2d y=\Lmb_0\|\psi\|_{L_2(\mathds{R}^2)}^2
\end{aligned}
\end{equation}
for each $\psi\in W_2^2(\mathds{R}^2)$. By the minimax principle it yields that the spectrum of $\Op$ lies in $[\Lmb_0,+\infty)$.

Assume that $\Lmb_0$ is a discrete eigenvalue of $\Op_0$ and $\psi_0(x)$ is the associated eigenfunction normalized in $L_2(\mathds{R})$. Let $\xi=\xi(y)$ be an infinitely differentiable cut-off function equalling one as $|y|<1$ and vanishing as $|y|>2$. It is easy to make sure that up to an appropriate renormalization the sequence of functions $(x,y)\mapsto \psi_0(x)\E^{\iu k y} \xi(y/m)$ is a characteristic one for $\Op$ at the point $\Lmb_0+k^2$ for each $k>0$.  If $\Lmb_0$ is the bottom of the essential spectrum of $\Op_0$, then for each $\Lmb>\Lmb_0$ there exists a characteristic sequence $u_m(x)$ for $\Op_0$ at point $\Lmb$. Then it is easy to check that an appropriate renormalization the sequence of functions $(x,y)\mapsto u_m(x)\xi(y/m)$ is a characteristic one for $\Op$ at point $\Lmb$. It remains to apply Weyl criterion to complete the proof.
\end{proof}

The rest of the section is devoted to the proof of Theorem~\ref{th2.2}.

We let $\chi(x,y):=e^{a\sqrt{x^2+y^2}}$ and by $L_2(\mathds{R}^2,\chi dxdy)$ we denote the weighted space of functions
$$
L_2(\mathds{R}^2,\chi dxdy):=\Big\{u\in L_2(\mathds{R}^2):\, \int\limits_{\mathds{R}^2} |u|^2\chi dxdy\Big\}.
$$

Let $\Lmb_j<\Lmb_e$ be the discrete eigenvalues of $\Op_0$ and $\psi_j(x)$ be the associated eigenfunctions normalized in $L_2(\mathds{R})$. Since operator $\Op_0$ is one-dimensional, all its discrete eigenvalues are simple.

We introduce a subspace $V\subset L_2(\mathds{R}^2)$ as a set of functions $f\in L_2(\mathds{R}^2)$ satisfying the condition
\begin{equation*}%\l%abel{3.7}
\int\limits_{\mathds{R}} f(x,y)\psi_j(x)dx=0\quad\text{for a.e.}\ y\in \mathds{R}\  \text{and for all}\ j.
\end{equation*}

Our next statement describes the behavior of the resolvent $(\Op-z)^{-1}f$ as $z$ is separated from $[\vk^2,+\infty)$.

\begin{lemma}\label{lm3.2}
Let $z\in\mathds{C}$, $\dist(z,[\vk^2,+\infty))\geqslant c_0$, where $c_0$ is a fixed positive constant. Then for each $f\in L_2(\mathds{R}^2,\chi dxdy)$ the equation
\begin{equation}\label{3.4}
(-\D+f'(\phi)-z)u=f\quad \text{in}\quad \mathds{R}^2
\end{equation}
is solvable and its solution can be represented as
\begin{equation}\label{3.5}
\begin{aligned}
u(x,y,z)=&-\sum\limits_{j} \frac{\psi_j(x)}{2k_j(z)}\int\limits_{\mathds{R}^2} \psi_j(t_1) \E^{-k_j(z)|y-t_2|}f(t_1,t_2)dt_1dt_2
\\
&+ (\Res(z)f)(x,y,z),
\end{aligned}
\end{equation}
where $k_j(z):=\sqrt{\Lmb_j-z}$, the branch of the square root is fixed by the condition $\sqrt{1}=1$,
and   operator $\Res(z): L_2(\mathds{R}^2,\chi dxdy)\mapsto W_2^2(\mathds{R}^2)$ is bounded and holomorphic w.r.t. $z$ such that $\dist\big(z,[\vk^2,+\infty)\big)\geqslant c_0$.
\end{lemma}

\begin{proof}
We represent $\g f$ as
\begin{equation*}
f(x,y)=\sum\limits_{j} \psi_j(x)f_j(y)+f^\bot(x,y),\quad f_j(y):=\int\limits_{\mathds{R}} f(x,y)\psi_j(x)dx.
\end{equation*}
It is clear that $f^\bot\in V\cap L_2(\mathds{R}^2,\chi dxdy)$, $f_j\in L_2(\mathds{R},e^{a|y|}dy)$, where weighted space $L_2(\mathds{R},e^{a|y|}dy)$ is introduced similar to $L_2(\mathds{R}^2,\chi dxdy)$.

Hence, we can solve problem (\ref{3.4}) independently for
for the right hand sides $f=\psi_j f_j$ and $f=f^\bot$. A solution $u_j$  to the latter can be found explicitly
\begin{equation*}%\l%abel{3.8}
u_j(x,y,z)=-\frac{\psi_j(x)}{2 k_j(z)} \int\limits_{\mathds{R}} \E^{-k_j(z)|y-t_2|}f_j(t_2)dt_2.
\end{equation*}
Functions $u_j$ are well-defined since $f_j\in L_2(\mathds{R},e^{a|y|}dy)$.

We seek the solution to problem (\ref{3.4}) with $f=f^\bot$ in the space $V\cap W_2^2(\mathds{R}^2)$. It is easy to check that space $V$ is an invariant one for operator $\Op$. Given arbitrary $\psi\in V\cap W_2^2(\mathds{R}^2)$, we employ  the definition of $\Lmb_0$ and  of space $V$ and proceed as in (\ref{3.9}):
\begin{equation*}
(\Op\psi,\psi)_{L_2(\mathds{R}^2)}=\|\nabla \psi\|_{L_2(\mathds{R}^2)}^2 +  (f'(\phi)\psi,\psi)_{L_2)(\mathds{R}^2)} \geqslant \Lmb_e\|\psi\|_{L_2(\mathds{R}^2)}^2.
\end{equation*}
By the minimax principle it follows that the spectrum for the restriction of $\Op$ on $V$ is $[\Lmb_e,+\infty)$. Thus, the operator $\Res(z):=(\Op-z)^{-1} : V\mapsto V\cap W_2^2(\mathds{R}^2)$ is well-defined and bounded. By the standard properties of the resolvent for a self-adjoint operator we conclude that $\Res$ is holomorphic w.r.t. $z$ once $\dist\big(z,[\Lmb_0,+\infty]\big)\geqslant c_0$.
\end{proof}

Let us discuss the latter lemma. For each $f\in L_2(\mathds{R}^2,\chi dxdy)$, the function $u=(\Op-z)^{-1}f$ solves equation (\ref{3.4}). At the same time, formula (\ref{3.5}) for solution to equation (\ref{3.4}) gives the function which belongs to $L_2(\mathds{R}^2)$ only as $\RE k_j(z)>0$, otherwise this function increases exponentially at infinity. In view of this fact, formula (\ref{3.5}) provides an analytic continuation of the resolvent $(\Op-z)^{-1}$ w.r.t. $z$.

As formula (\ref{3.5})  states, the solution to equation (\ref{3.4}) has singularities w.r.t. $z$ at points $\Lmb_j$.  Moreover, let $z=\Lmb_j-k^2$, where $k$ is small enough. Then by Lemma~\ref{lm3.2} the solution to problem (\ref{3.4}) can be represented as
\begin{equation}\label{3.10}
u(x,y,k)=-\frac{\psi_j(x)}{2k}\int\limits_{\mathds{R}^2} \psi_j(t_1)\E^{-k|y-t_2|}f(t_1,t_2)dt_1 dt_2 + \Res_j(k)f,
\end{equation}
where $\Res_j(k): L_2(\mathds{R}^2)\mapsto W_2^2(\mathds{R}^2)$ is a bounded operator holomorphic w.r.t. sufficiently small $k$. Since $f$ is compactly supported, we can expand the exponential in (\ref{3.10}) w.r.t. $k$ and rewrite then (\ref{3.10}) as
\begin{equation}\label{3.11}
u(x,y,k)=-\frac{\ell_j(f)}{2k}\psi_j(x)+\widetilde{\Res}_j(k)f,
\quad \ell_j(f):=\int\limits_{\mathds{R}^2} \psi_j(t_1)f(t_1,t_2)dt_1d t_2,
\end{equation}
where
\begin{align*}%\l%abel{3.7a}
(\widetilde{\Res}_j(k)f)(x,y)= (\Res_j(k)f)(x,y)
%\\
%&
- \psi_j(x) \int\limits_{\mathds{R}^2} (e^{-k|y-t_2}-1) \psi_j(t_1) f(t_1,t_2)dt_1dt_2.
\end{align*}
In view of the above formula and the properties of $\Res_j(k)$ we see that $\widetilde{\Res}_j(k)$ is a bounded operator from $L_2(\mathds{R}^2,\chi dxdy)$ into $W_2^2(\mathds{R}^2)$ holomorphic w.r.t. small $k$.

We proceed to equation (\ref{2.5}). To study it, we apply a modified version of non-self-adjoint Birman-Schwinger principle suggested  in \cite{Ga}, see also \cite{MSb06}, with certain modifications. The modifications are needed since here we deal with operator pencil instead of classic eigenvalue equation.

%We rewrite (\ref{2.5}) as
%\begin{equation}\label{3.12}
%(-\D +f'(\phi)+\lmb^2)\psi+\e\lmb\g\psi=0.
%\end{equation}
Assume that
\begin{equation}\label{3.15}
\dist\big(-\lmb^2, \spec(\Op_0)\big)>c_1,
\end{equation}
where $c_1$ is a fixed positive constant,
and $\lmb$ ranges in a compact set. Then we can apply the resolvent $(\Op+\lmb^2)^{-1}$ to (\ref{2.5}) and get
\begin{equation*}%\l%abel{3.17}
\big(\I+\e\lmb(\Op+\lmb^2)^{-1}\g\big)\psi=0.
\end{equation*}
In accordance with Lemma~\ref{lm3.2}, for the considered values of $\lmb$, the operator $(\Op+\lmb^2)^{-1}$ is bounded uniformly in $\lmb$. Hence, we can invert the operator in the left hand side of the latter equation and it leads us to $\psi=0$. Thus, provided the first inequality in (\ref{3.15}) is obeyed, and $\lmb^2$ ranges in a compact set, equation (\ref{2.5}) can have non-trivial solutions only as $-\lmb^2$ is close to $\Lmb_j$. This is why in what follows we consider only small neighborhoods of $\Lmb_j$.

Assume that
\begin{equation}\label{3.16}
\lmb^2(k)=k^2-\Lmb_*,\quad  \lmb(k)=\pm\iu\sqrt{\vk_*^2-k^2}, \quad \vk_*:=\sqrt{\Lmb_*},
\end{equation}
where $k$ is complex and sufficiently small. In what follows we regard $k$ as a new spectral parameter.

We rewrite equation (\ref{2.5}) as
\begin{equation}\label{3.23}
(-\D+f'(\phi)+\lmb^2)\psi=-\e\lmb\g\psi.
\end{equation}
Then we can invert the operator in the left hand side and by (\ref{3.10}) we obtain
\begin{equation}
\psi(x,y)= -\frac{\ell_*(f_\e)}{2k}\psi_*(x) + (\widetilde{\Res}_*(k)(f_\e))(x,y),\label{3.18}
\end{equation}
where $\ell_j$ and $\widetilde{\Res}_*(k)$   denote $\ell_j$ and  $\widetilde{\Res}_j(k)$ associated with eigenvalue $\Lmb_*$, and
\begin{equation*}
f_\e:=-\e\lmb(k) \g\psi\in L_2(\mathds{R}^2,\chi dxdy).
\end{equation*}
We substitute this identity into (\ref{3.23}) and arrive at the equation for $f_\e$:
\begin{equation*}
f_\e=\frac{\e\lmb(k)}{2k}\ell_*(f_\e)\g\psi_*-\e\lmb\g \widetilde{\Res}_*(k)f_\e.
\end{equation*}
In view of the above described properties of operator $\widetilde{\Res}_*$, operator $\g \widetilde{\Res}_*$ is a bounded one in $L_2(\mathds{R}^2,\chi dxdy)$ holomorphic w.r.t. small $k$. Hence, the operator
$\mathcal{L}_*(\e,k):=\big(\I+\e\lmb(k)\g\widetilde{\Res}_*(k)\big)^{-1}$ is well-defined being a bounded operator in $L_2(\mathds{R}^2,\chi dxdy)$ holomorphic w.r.t. $k$. The first term of the Neumann series for $\mathcal{L}_*$ reads as
\begin{equation}\label{3.14}
\mathcal{L}_*(\e,k)=\I- \e\lmb(k)\g\widetilde{\Res}_*(k)+\Odr(\e^2)= \I-\e\lmb(0)\g\widetilde{\Res}_*(0)+\Odr(\e^2+\e |k|).
\end{equation}

We move the second term in the right hand side of (\ref{3.18}) and then we apply   operator $\mathcal{L}_*(\e,k)$ to the equation. It implies
\begin{equation}\label{3.19}
f_\e=\frac{\e\lmb(k)}{2k}\ell_*(f_\e)\mathcal{L}_*(\e,k)(\g\psi_*).
\end{equation}
We observe that $\g\psi_*\in L_2(\mathds{R}^2,\chi dxdy)$, so, the right hand side is well-defined. Moreover, $\ell_*(\g\psi)$ can not be zero, since otherwise the latter equation implies $\psi=0$, while we are interesting in non-trivial solutions to (\ref{2.5}). Taking into consideration this fact, we apply functional $\ell_*$ to (\ref{3.19}) and cancel out the term $\ell_*(\g\psi)$. It leads us to the final equation:
\begin{equation}\label{3.6}
2k=\e F(\e,k),\quad F(\e,k):=\e\lmb(k)\ell_*\big(\g\mathcal{L}_*(\e,k)\g\psi_*\big).
\end{equation}
This is the equation for the values of $k$ for which equation (\ref{2.5}) has a nontrivial solution with $\lmb$ defined by (\ref{3.16}). We note that, in fact, we deal with two independent equations in (\ref{3.6}) related to two different branches of square root in (\ref{3.16}). As one can easily make sure by using (\ref{3.19}), the associated non-trivial solution to (\ref{3.19}) reads as
\begin{equation*}%\l%abel{3.20}
f_\e=\mathcal{L}_*(\e,k)(\g\psi_*),
\end{equation*}
where $k$ is a root to (\ref{3.6}). We observe that thanks to the above described properties of operator $\mathcal{L}_*$, function $F$ is holomorphic w.r.t. $k$.

Let us prove that for sufficiently small $\e$ equation (\ref{3.6}) has the unique solution $k=k_\e$ converging to zero as $\e\to0$. We take a small fixed $\d>0$ so that function $F$ is holomorphic in $B_\d:=\{k: |k|<\d\}$. Then it is bounded uniformly in $\e$ and $k$ and thus $|\e F|\leqslant |k|$ as $|k|=\d$ once $\e$ is small enough. Hence, by Rouch\'e theorem the function $k\mapsto 2k-\e F(\e,k)$ has the same amount zeroes (counting orders) in $B_\d$ as the function $k\mapsto 2k$ does. Thus, equation (\ref{3.6}) has the unique zero $k=k_\e$ in $B_\d$. This root converges to zero as $\e\to0$ that is implied immediately by equation (\ref{3.6}) since its right hand side tends to zero as $\e\to0$.

Consider the case $\Lmb_*=0$. Then $\lmb(k)=\pm k$ and in this case equation (\ref{3.6}) has the only solution $k_\e=0$.
It means that the perturbation $\e\lmb\g$ does not change this spectral point of the unperturbed system. It completes the proof of Statement 1 in Theorem~\ref{th2.2}.

Consider the case $\Lmb_*\not=0$ and let us calculate the asymptotics for $k_\e$ as $\e\to0$.  We substitute (\ref{3.14}) into the definition of $F$:
\begin{equation*}
F(\e,k)=\e\lmb(0)\ell_*(\g\psi_*)-\e^2\lmb^2(0) \ell_*\big(\g \widetilde{\Res}_*(0) \psi_*\big) + \Odr(\e|k|^2+\e^2|k|+\e^3).
\end{equation*}
This formula and equation (\ref{3.6}) implies the asymptotics for $k_\e$:
\begin{align*}%\l%abel{3.21}
&k_\e=\pm\iu\e\vk_* K_1 + \e^2\vk_*^2 K_2+\Odr(\e^3),
\\
&K_1:= \frac{1}{2}\ell_*( \psi_*), \quad K_2:= \frac{1}{2} \ell_*( \g  \widetilde{\Res}_*(0)\g\psi_*).
\end{align*}

We denote $U_*:=\widetilde{\Res}_*(0)(\g\psi_*)$. Let us check that it satisfies (\ref{2.16}), (\ref{2.17}). The second term in (\ref{3.10}) with $f=\g\psi_*$ belongs to $W_2^2(\mathds{R}^2)$  for all $k$ including $k=0$, while the first term  can be expanded as
\begin{equation}\label{3.2}
\begin{aligned}
\frac{\psi_*(x)}{2k} &\int\limits_{\mathds{R}^2} \psi_*(t_1)\E^{-k|y-t_2|}\g(t_1,t_2)\psi_*(t_1)dt_1 d t_2
\\
&=-\frac{\ell_*(f)}{2k}\psi_*
+ \frac{\psi_*(x)}{2}  \int\limits_{\mathds{R}^2}|y-t_2|  \psi_*(t_1)\g(t_1,t_2) dt_1 dt_2 + O(k),\quad k\to0,
\end{aligned}
\end{equation}
and
\begin{equation*}
(-\D+f'(\phi)-\Lmb_*^2)\psi_*(x)\int\limits_{\mathds{R}^2} |y-t_2|\psi_*(t_1)\g(t_1,t_2) dt_1 dt_2=\ell_*(\g\psi_*)\psi_*(x).
\end{equation*}
We substitute this representation into (\ref{3.10}) and the result is substituted then into the corresponding equation (\ref{3.4}). Then we pass to the limit as $k\to0$ and arrive at equation (\ref{2.16}). Asymptotics (\ref{2.17}) is implied immediately by the definition of $\widetilde{\Res}_*$ and (\ref{3.2}). Thus, together with  (\ref{3.11}), it yields that the formulae for $K_j$ can be rewritten as (\ref{3.24}).

If  $\RE k_\e>0$,  it follows from (\ref{3.5}) the associated solution to problem (\ref{3.23}) decays exponentially at infinity. Thus, it is an eigenfunction and $\lmb(k_\e)$ is the corresponding eigenvalue. Once $\RE k_\e<0$, we deal with a resonance. To check the sign of $\RE k_\e$, it is sufficient to find the sign of $\RE(\pm\iu\e\vk_* K_1 + \e^2\vk_*^2 K_2)$. And Statements~2-6 in Theorem~\ref{th2.2} appear just as various particular cases of calculating the sign of $\RE(\pm\iu\e\vk_* K_1 + \e^2\vk_*^2 K_2)$ for various combinations of signs  of $\Lmb_*$, $K_1$, and $K_2$.
The proof of Theorem~\ref{th2.2} is complete.

%{\itshape Figures:} Figures should be cited in the main text in chronological order. This is dummy text with a citation to the first figure (\textbf{Figure \ref{fig1}}). Only the first citation in the main text of  Figure \ref{fig1}, or any other display item, should be bolded. Citations to Figure \ref{fig1} later in the text should not be bolded.

%\begin{figure}
%\centerline{\includegraphics[width=3in]{art/mouse.eps}}
%\caption{Figure Caption.}
%\label{fig1}
%\end{figure}

\section{Conclusion} \label{Sec4}
The eigenvalue problem for the dynamical perturbation around the static kink solution to the 2D Klein-Gordon continuum equation perturbed by a spatially localized, $\mathcal{PT}$-symmetric term has been analysed. The main result is presented in Theorem 1 in Sec. \ref{Sec2}, which describes the behavior of the eigenvalues corresponding to the extended and spatially localized eigenmodes as the functions of small perturbation parameter.

% Acknowledgements
\section*{Acknowledgments}
The is supported by the Russian Foundation for Basic Research, grant no. 15-31-20037.

%\begin{appendix}
%\section*{Appendix}
%Appendix text goes here.
%\begin{equation}
%x^2+y^2=1.
%\label{eq:appeqa}
%\end{equation}
%
%
%\end{appendix}

% References


\begin{thebibliography}{00}

\bibitem{J2025}
\textsc{P.~J.~Johnnon}, \emph{Numerical Methods for Partial
Differential Equations} (SIAM, Philadelphia, 2025).

\bibitem{Bender1}
\textsc{C.~M.~Bender} and \textsc{S.~Boettcher}, Real spectra in nonhermitian
hamiltonians having PT symmetry, \emph{Phys. Rev. Lett.} 80: 5243--5246 (1998).

\bibitem{Bender2}
\textsc{C.~M.~Bender}, \textsc{D.~C.~Brody}, and \textsc{H.~F.~Jones}, Complex extension of quantum mechanics, Phys. Rev. Lett. 89, 270401–4 (2002).

\bibitem{Bender3}
\textsc{C.~M.~Bender}, Making Sense of Non-Hermitian Hamiltonians,
Rep. Prog. Phys. 70, 947-1018 (2007).

\bibitem{QM1}
\textsc{K. Jones-Smith} and \textsc{H.~Mathur}, Non-Hermitian quantum Hamiltonians with PT symmetry, Phys. Rev. A 82, 042101 (2010).

\bibitem{QM2}
\textsc{C.~M.~Bender} and \textsc{P.~D.~Mannheim}, PT symmetry in relativistic quantum mechanics, Phys. Rev. D 84, 105038 (2011).

\bibitem{Optics1}
\textsc{R.~El~Ganainy}, \textsc{K.~G.~Makris}, \textsc{D.~N.~Christodoulides} and \textsc{Z.~H.~Musslimani}, Theory of coupled optical PTsymmetric structures, Opt. Lett. 32, 2632–2634 (2007).

\bibitem{Optics2}
\textsc{C.~E.~Ruter}, \textsc{K.~G.~Makris}, \textsc{R.~El~Ganainy}, \textsc{D.~N.~Christodoulides}, \textsc{M.~Segev} and \textsc{D.~Kip}, Observation of parity-time symmetry in optics, Nat. Phys. 6, 192–195 (2010).

\bibitem{Optics3}
\textsc{S.~V.~Suchkov}, \textsc{A.~A.~Sukhorukov}, \textsc{J.~Huang}, \textsc{S.~V.~Dmitriev}, \textsc{C.~Lee} and \textsc{Yu.~S.~Kivshar}, Nonlinear switching and solitons in PT-symmetric photonic systems, arXiv:1509.03378v1

\bibitem{Electricity1}
\textsc{J.~Schindler}, \textsc{A.~Li}, \textsc{M.~C.~Zheng}, \textsc{F.~M.~Ellis} and \textsc{T.~Kottos}, Phys. Rev. A 84, 040101 (2011).

\bibitem{Electricity2}
\textsc{H.~Ramezani}, \textsc{J.~Schindler}, \textsc{F.~M.~Ellis}, \textsc{U.~Gunther} and \textsc{T.~Kottos}, Bypassing the bandwidth theorem with PT symmetry, Phys. Rev. A 85, 062122 (2012).

\bibitem{Mechanics}
\textsc{C.~M.~Bender}, \textsc{B.~K.~Berntson}, \textsc{D.~Parker} and \textsc{E.~Samuel}, Observation of PT phase transition in a simple mechanical system, Am. J. Phys. 81, 173 (2013).

\bibitem{Mechanics1}
\textsc{J.~Cuevas}, \textsc{P.~G.~Kevrekidis}, \textsc{A.~Saxena} and \textsc{A.~Khare}, PT-symmetric dimer of coupled nonlinear oscillators, Phys. Rev. A 88, 032108 (2013).

\bibitem{KG1}
\textsc{A.~Demirkaya}, \textsc{D.~J.~Frantzeskakis}, \textsc{P.~G.~Kevrekidis}, \textsc{A.~Saxena} and \textsc{A.~Stefanov}, Effects of parity-time symmetry in nonlinear Klein-Gordon models and their stationary kinks, Phys. Rev. E 88, 023203 (2013).

\bibitem{DKKSS}
\textsc{A.~Demirkaya}, \textsc{T.~Kapitula}, \textsc{P.~G.~Kevrekidis}, \textsc{M.~Stanislavova} and \textsc{A.~Stefanov}, Stud. Appl. Math. 133, 298–317 (2014).

\bibitem{Lu}
\textsc{N.~Lu}, \textsc{P.~G.~Kevrekidis} and \textsc{J.~Cuevas-Maraver}, PT-symmetric sine-Gordon breathers, J. Phys. A: Math. Theor. 47, 455101 (2014).

\bibitem{DemirkayaPLA}
\textsc{A.~Demirkaya} and \textsc{S.~Hakkaev}, On the spectral stability of periodic waves of the coupled Schr\"odinger equations, Phys. Lett. A 379, 2908-2914 (2015).

\bibitem{K1}
\textsc{P.~G.~Kevrekidis}, Variational method for nonconservative field theories: Formulation and two PT-symmetric case examples, Phys. Rev. A 89, 010102(R) (2014).

\bibitem{Danial1}
\textsc{D.~Saadatmand}, \textsc{S.~V.~Dmitriev}, \textsc{D.~I.~Borisov} and \textsc{P.~G.~Kevrekidis}, Interaction of sine-Gordon kinks and breathers with a parity-time-symmetric defect, Phys. Rev. E 90, 052902 (2014).

\bibitem{Danial2}
\textsc{D.~Saadatmand}, \textsc{S.~V.~Dmitriev}, \textsc{D.~I.~Borisov}, \textsc{P.~G.~Kevrekidis}, \textsc{M.~A.~Fatykhov} and \textsc{K.~Javidan}, Effect of the phi-4 kink's internal mode at scattering on a PT-symmetric defect, JETP Lett. 101, 550-555 (2015).

\bibitem{Danial3}
\textsc{D.~Saadatmand}, \textsc{S.~V.~Dmitriev}, \textsc{D.~I.~Borisovd}, \textsc{P.~G.~Kevrekidis}, \textsc{M.~A.~Fatykhov} and \textsc{K.~Javidan}, Kink scattering from a parity-time-symmetric defect in the phi-4 model, Commun. Nonlinear. Sci. Numer. Simulat. 29, 267–282 (2015).

\bibitem{Bi} \textsc{M.~S.~Birman}, Perturbation of the continuous spectrum of a singular elliptic operator under a change of the boundary and the boundary condition, Vestnik Leningradskogo universiteta. 1, 22-55 (1962).

\bibitem{Ga} \textsc{R.~R.~Gadylshin}, On local perturbations of the Schrodinger operator on the axis, Theor. Math. Phys. 132(1), 976-982 (2002).

\bibitem{MSb06} \textsc{D.~Borisov}, Discrete spectrum of an asymmetric pair of waveguides coupled through a window, Mat. Sb. 197(4), 3-32 (2006). [English translation: Sbornik: Mathematics 197(4):475-504, 2006.]

\end{thebibliography}
\end{document}